\documentclass[11pt]{article}

\usepackage{amssymb}
\usepackage{amsmath}
\usepackage{amsfonts}
\usepackage{amsthm}
\usepackage{boxedminipage}
\usepackage[margin=1.5in]{geometry}
\usepackage[colorinlistoftodos]{todonotes}
\usepackage{comment}

\usepackage{tikz} 
\usetikzlibrary{matrix}

\usepackage{color,graphicx}
\title{Partially Thermostated Kac Model}
\author{\vspace{5pt} Hagop Tossounian$^1$, Ranjini Vaidyanathan$^1$ \\ \vspace{5pt} \small{$^1$School of Mathematics, Georgia Institute of Technology}\\[-6pt]
\small{$686$ Cherry Street Atlanta, GA $30332-0160$ USA}}

\begin{document}

\pagestyle{plain}
\graphicspath{{Graphics/}}

\newtheorem{thm}{Theorem}[section]
\newtheorem{cor}[thm]{Corollary}
\newtheorem{lem}[thm]{Lemma}
\newtheorem{prop}[thm]{Proposition}

\theoremstyle{definition}
\newtheorem{defin}[thm]{Definition}
\newtheorem{rem}[thm]{Remark}
\newtheorem*{xrems}{Remarks}
\newtheorem{exa}[thm]{Example}

\newtheorem*{xrem}{Remark}

\newcommand{\V}{\ensuremath{\mathbf{v}}}
\newcommand{\el}{\ensuremath{\mathcal{L}}}
\newcommand{\RR}{\mathbb{R}}
\def\<#1>{\langle#1\rangle}

\def\Xint#1{\mathchoice
{\XXint\displaystyle\textstyle{#1}}%
{\XXint\textstyle\scriptstyle{#1}}%
{\XXint\scriptstyle\scriptscriptstyle{#1}}%
{\XXint\scriptscriptstyle\scriptscriptstyle{#1}}%
\!\int}
\def\XXint#1#2#3{{\setbox0=\hbox{$#1{#2#3}{\int}$ }
\vcenter{\hbox{$#2#3$ }}\kern-.6\wd0}}
\def\ddashint{\Xint=}
\def\dashint{\Xint-}

\maketitle

{\let\thefootnote\relax\footnote{Work partially supported by U.S. National Science Foundation grant DMS 1301555 \\ \copyright\, 2015 by the authors. This paper may be reproduced, in its entirety, for non-commercial purposes.}}

\begin{abstract}
We study a system of $N$ particles interacting through the Kac collision, with $m$ of them interacting, in addition, with a Maxwellian thermostat at temperature $\frac{1}{\beta}$. We use two indicators to understand the approach to the equilibrium Gaussian state. We prove that i) the spectral gap of the evolution operator behaves as $\frac{m}{N}$ for large $N$ ii) the relative entropy approaches its equilibrium value (at least) at an eventually exponential rate $\sim \frac{m}{N^2}$ for large $N$. The question of having non-zero entropy production at time $0$ remains open. A relationship between the Maxwellian thermostat and the thermostat used in \cite{BLV} is established through a van Hove limit.
\end{abstract}

\section{Introduction}

Mark Kac introduced a stochastic model of $N$ identical particles interacting through binary collisions \cite{Kac}. The particles are constrained in $1$ dimension and are uniformly distributed in space. Hence, the phase space consists of 1-dimensional velocities $\V=(v_1,...,v_N)$ that evolve when the particles undergo random collisions as follows: Two particles $i,j$ are chosen uniformly among the $\binom{N}{2}$ pairs, and $\theta \in [0,2\pi)$ is chosen uniformly. The outgoing velocities $v_i^\ast$ and $v_j^\ast$ are given by $v_i \cos{\theta}+v_j \sin{\theta}$ and $-v_i \sin{\theta}+v_j \cos{\theta}$ respectively, where $v_i$ and $v_j$ are the incoming velocities of particles $i$ and $j$.

This collision preserves the kinetic energy and hence $\V$ lies on the constant energy sphere $S^{N-1}(\sqrt{NE})$, where $E$, the energy per particle, is determined by the initial condition. The system is modeled as a Markov jump process with collision times that are exponentially distributed with mean $\frac{1}{N\lambda}$. A probability density $f(\V,t)$ on the phase space evolves through the corresponding Kolmogorov forward equation, called the Kac master equation:

\begin{equation} \label{E:Kmast}
\frac{\partial f}{\partial t} = N\lambda (Q-I)f 
\end{equation}

where $Q=\frac{1}{\binom{N}{2}} \displaystyle\sum_{i<j} Q_{ij}$ is the Kac operator with 

\[Q_{ij}f=\dashint_0^{2\pi} f(...,v_i \cos{\theta}+v_j \sin{\theta},...,-v_i \sin{\theta}+v_j \cos{\theta},...) \,d\theta\]

The unique equilibrium state is the uniform distribution on the sphere.

In his paper, Kac precisely formulates Boltzmann's Stosszahlansatz (molecular chaos hypothesis), which says that for a dilute gas in the large particle limit, the incoming velocities of colliding particles are uncorrelated. Kac proved for his model that this property propagates in time. This notion, now known as ``propagation of chaos'', enabled him to rigorously derive a space{-}homogeneous Boltzmann equation for his model \cite{Kac} (see also \cite{McK}). In fact, one of Kac's motivations was to study approach to equilibrium for the Boltzmann equation using the \emph{linear} $N$ particle master equation \eqref{E:Kmast}. In particular, Kac conjectured that the spectral gap of $N\lambda (I-Q)$ is bounded below, away from 0, uniformly in $N$. This was proved by Janvresse \cite{Janvresse} and the exact gap was found by Carlen, Carvalho, and Loss in \cite{CCL} and independently by Maslen \cite{Maslen}. It follows from their work that $\vert\vert f(\vec{v},t) - 1 \vert\vert_2 \leq e^{-\frac{\lambda}{2}t} \vert\vert f(\vec{v},0) - 1 \vert\vert_2 $, where $f(\vec{v},t)$ satisfies \eqref{E:Kmast}, and the norm is in $L^2(S^{N-1}(\sqrt{NE}),d\sigma)$, where $d\sigma$ is the normalized uniform measure on the sphere.


It turns out that the relative entropy $S(f\vert 1) = \int f \log f d\sigma$, being an extensive quantity, is a more favorable measure of distance to equilibrium in the large particle limit. For the Kac model, entropic approach to the equilibrium at an exponential rate of order $\frac{1}{N}$ is shown by Villani in \cite{Villani}. This rate was shown to be essentially optimal near $t=0$ by Einav in \cite{amit}, by  constructing states in which a macroscopic fraction of the kinetic energy was contained in a fraction $\sim N^{-\alpha}$ of the particles, for $\alpha>0$ suitably chosen.

The Kac model coupled to a heat bath was studied in \cite{BLV}, where they explored the possibility of obtaining better entropic convergence by remaining close to physically realizable initial states. To this end,  they considered a system of $N$ particles where, in addition to the Kac collision among them, each collides with a reservoir modeled as an infinite gas in thermal equilibrium. This resulted in a system in which all except a relatively small number of particles are in equilibrium. Exponential convergence to the canonical equilibrium at a rate of $\frac{\mu}{2}$ was proved, where $\mu$ is the strength of the thermostat. Note that since the energy of the $N$ particles is not conserved in the presence of a heat bath, the phase space becomes $\RR^N$.

In this paper, we take the model in $\cite{BLV}$ but let only $m<N$ of the particles 
be thermostated, and use a simpler model for the thermostat: the Maxwellian thermostat given by \eqref{E:strongt}. (We will refer to the Maxwellian thermostat as the strong thermostat, and to the thermostat used in \cite{BLV} (see \eqref{E:weak}) as the weak thermostat.) The motivation for our study is two-fold. First, studying partially thermostated systems is a step towards introducing spatial inhomogeneity in Kac-type models, by viewing the $m$ thermostated particles as situated ``closer'' to the heat bath. These $m$ particles act as the medium of heat exchange between the other particles and the reservoir. Second, the convergence to equilibrium in \cite{BLV} persisted even without the interparticle interaction, which did not play a role in the slowest decay modes. By thermostating only a subset of the particles, the interparticle interaction become necessary for the system to approach the canonical equilibrium and hence their role can be better understood. Using the spectral gap, we show that (strongly) thermostating a macroscopic fraction of particles i.e. $m = \alpha N$ guarantees approach to equilibrium in the $L^2$ distance uniformly in $N$. We also obtain a weaker convergence result in terms of the relative entropy of the system.

\section*{Description of Model and Results}

We have $N$ particles interacting via the Kac collision, with $m$ among them interacting, in addition, with a Maxwellian thermostat at inverse temperature $\beta$. We fix $N \geq 2$, $1 \leq m < N$ (The case $m=N$ has been studied in \cite{BLV}, using the weak thermostat.) When particle $k \in \{1,...,m\}$ is thermostated, it forgets its precollisional velocity $v_k$ and is given a new velocity from the Gaussian distribution at the temperature of the heat bath. Physically, this could model the behavior of a particle colliding a large number of times with particles from the heat bath. This can also be thought of as a particle in the system being replaced with one from the heat bath. 

The collision times with the heat bath for particles $\{1, \dots, m\}$ are independent and exponentially distributed with parameter $\mu$. The master equation for the evolution of a phase space probability density $f(\V,t)$ is given by

\begin{equation} \label{E:mast}
\frac{\partial f}{\partial t}= N \lambda (Q-I)f + \mu \sum_{k=1}^m(R_k - I)f \,\,\,\, ,
\end{equation}

where the operator

\begin{equation}\label{E:strongt}
R_k f := \sqrt{\frac{\beta}{2\pi}} e^{-\beta \frac{v_k^2}{2}} \int{dw f(v_1,v_2,...,v_{k{-}1},w,v_{k+1},\dots,v_N)}
\end{equation}

corresponds to the thermostat acting on the $k$-th particle. Recall that phase space is $\mathbb{R}^N$ since our system is non-isolated and energy is not conserved.  
We assume that the particles $1,2, \dots,m $ and the particles $m+1,...,N$ are indistinguishable, i.e. $f(\V,t)$ is symmetric under exchange of variables $v_1, \dots, v_m$ and under the exchange of variables $v_{m+1},..., v_N$. The evolution preserves this symmetry. The interplay between the thermostat-interaction and the Kac collisions that distribute the energy to the non-thermostated particles lead the system to equilibrium. 

As we will see, the unique equilibrium of eq.~\eqref{E:mast} is the Gaussian

\[\gamma(\V):= \prod_{k=1}^{N} g(v_k) :=\prod_{k=1}^{N} \sqrt{\frac{\beta}{2\pi}} e^{-\beta \frac{v_k^2}{2}} \,\, . \]
The evolution operator in eq.~\eqref{E:mast} is not self-adjoint on $L^2(\mathbb{R}^N)$, and to this end we make the ground-state transformation

\begin{equation}
f(\V)=\gamma(\V) \bigl( 1+h(\V)\bigr) \,\,\, , \label{E:gs}
\end{equation}

where $\int h \gamma =0$. The evolution equation for $h(\V,t)$ becomes

\begin{equation} \label{E:hmast}
\frac{\partial h}{\partial t}= N \lambda (Q-I)h + \mu \sum_{k=1}^m(P_k - I)h \,\,\,\, ,
\end{equation}

where 

\[ P_k h:= \int{dw g(w) h(v_1,...,v_{k-1},w,v_{k+1},\dots,v_N)} \]

is a function independent of $v_k$. In the Hilbert space $L^2(\mathbb{R}^N,\gamma)$, the operators $P_k$, $Q$ and hence $\el_{N,m}:= N \lambda (I-Q) + \mu \sum_{k=1}^m(I-P_k)$ associated with the evolution, are self-adjoint. In fact, each $P_k$ is a projection.

The rate at which $h$ tends to its equilibrium value $0$ in the space $L^2(\mathbb{R}^N,\gamma)$ is given by the spectral gap $\Delta_{N,m}$ (see \eqref{E:gapdef}). Theorem~\ref{T:gap} states that $\Delta_{N,m} \sim \frac{m}{N}$ for large $N$. It turns out that the kinetic energy $K(t):=\int (\sum_{j=1}^N v_j^2) f(\V,t) d\V$ also behaves similarly for large $N$. More precisely, $K(t)$, which is not conserved since the $N$-particle system is not isolated, tends to its equilibrium value $\frac{N}{2\beta}$ at a rate $\sim \frac{m}{N}$ when $N$ is large. 

\begin{xrem}
The behavior of the kinetic energy is indicative of the action of the operator $\el_{N,m}$ on polynomials of the form $v_j^2$. Moreover, for $N=2, m=1$, we show in Appendix \ref{A:sp2} that the gap eigenfunction - the slowest rate of decay in the space $L^2(\mathbb{R}^2,\gamma)$ - is a second degree polynomial. One may thus wonder if the gap eigenfunction is a second degree polynomial for other values of $N$ too. However, currently we only have asymptotic bounds on $\Delta_{N,m}$.\end{xrem}

Next, we study the behavior of the relative entropy $S(f|\gamma)$ (defined in \eqref{E:relent}). To obtain a quantitative rate for the decay in the relative entropy (we use the opposite sign for the relative entropy), one could try to prove Cercignani's conjecture \cite{Cerc} applied to our system:

\begin{equation} \label{E:prod}
-\frac{dS(f(.,t)|\gamma)}{dt} \geq k S(f(.,0)|\gamma) 
\end{equation}

for some $k>0$, which would yield an exponential bound \[S(f(.,t)|\gamma) \leq e^{-kt} S(f(.,0)|\gamma)\] for the entropy. The quantity $-\frac{dS}{dt}$ is called the entropy production. Parenthetically, note that the spectral gap imposes a condition on how big $k$ can be: linearizing~\eqref{E:prod} and comparing lowest order terms gives 

\begin{equation} \label{E:entbd}
k \leq 2\Delta_{N,m}.\end{equation}

For our problem, finding a bound for the entropy production appears to be hard because the familiar methods to obtain such estimates fail (we demonstrate why at the end of Section \ref{S:entr}). Instead, we show in Theorem \ref{T:duh} that the entropy at time $t$ satisfies

\begin{equation} \label{E:entrN}
S(f(.,t)|\gamma) \leq D_{N,m}(t) S(f(.,0)|\gamma)
\end{equation}

where $f(.,t)$ is the solution of \eqref{E:mast} with initial condition $f(.,0)$ and 
\[D_{N,m}(t)=\left(-\frac{\delta_- e^{-\delta_+t}}{\delta_{+}-\delta_{-}} + \frac{\delta_+ e^{-\delta_-t}}{\delta_{+}-\delta_{-}} \right).\]

For large $N$ and $t$, $D_{N,m}(t)\sim \exp(-\frac{m\lambda \mu}{(N-1)(N\lambda+\mu)}t)$. Note that \eqref{E:entrN} is weaker than \eqref{E:prod}. For instance, \eqref{E:entrN} does not yield an entropy production bound at time $0$, since $D_{N,m}'(0)=0$. We prove the Theorem by employing the convexity of $S(f|\gamma)$ directly. The idea is similar to a method used in \cite{BLV} to study the entropy of a particle acted on by the weak thermostat.

The generator $U$ of the weak thermostat is defined as follows:

\begin{equation} \label{E:weak}
U [f(v)]:= \int dw\dashint_0^{2\pi} d\theta f(v \cos{\theta}+w\sin{\theta}) g(-v\sin \theta +w \cos \theta) \end{equation}

where $g (v)=\sqrt{\frac{\beta}{2\pi}} e^{-\beta v^2/2}$. 

The following entropy decay bound for the process is shown in \cite{BLV}:

 \begin{equation} \label{E:entrweak}
 S(e^{\eta (U-I)t}f|g) \leq e^{-\eta t/2} S(f|g) \, , \mbox{ or }
 \end{equation}

\begin{equation} \label{E:kappa}
\frac{dS}{dt} \leq -\frac{\eta}{2}S. 
\end{equation}

As an aside, we show in Appendix~\ref{A:appopt} that the bound in \eqref{E:entrweak} is optimal by using an optimizing sequence similar to that used in \cite{bobcerc,CCLRV,amit}. 

One can interpret the weak thermostat as a particle interacting with an infinite heat bath at temperature $\frac{1}{\beta}$ via the Kac-collision. The velocity distribution $g(v)$ of the particles in the heat bath is not affected by the collisions by virtue of the infinite size of the bath. 
This picture shows why it is weaker than the strong thermostat: In order for a particle from the system to forget its incoming velocity and pick a new one from the distribution $g(v)$, it has to undergo a large number of weak-thermostat interactions. The strong thermostat achieves this in one step \eqref{E:strongt}. Although the weak thermostat mimics heat bath interactions more naturally, the strong thermostat is advantageous to use as a first step since the corresponding operator is idempotent and thus mathematically simpler. Moreover, we demonstrate in Theorem \ref{T:vh} that the weak thermostat can be obtained from the strong thermostat via a van Hove limit.

The paper is organized as follows: We show approach to equilibrium in $L^2$ in Section~\ref{S:L2}, compute the van Hove limit in Section \ref{S:vh}, and show approach to equilibrium in relative entropy in  Section~\ref{S:entr}.

\section{Approach to equilibrium in $L^2$} \label{S:L2}

From this point on, we set $\beta=1$ without loss of generality.

In concurrence with the ground-state transformation \eqref{E:gs}, let $\mathcal{X}_N:=\{ u \in L^2(\mathbb{R}^N,\gamma) :  \<u,1> =0 \}$, where $\<.,.>$ denotes the inner product in the $L^2$ space with weight $\gamma$. The condition $\<u,1>=\int{u(\V) \gamma(\V) d\V}=0$ corresponds to the normalization of the probability density $f$. 

\begin{lem} \label{L:ker}\
\begin{itemize} 
\item $\el_{N,m} \geq 0$ on $\mathcal{X}_N$.
\item $\el_{N,m} h=0 \Leftrightarrow h=0$.
\end{itemize}
\end{lem}

\begin{proof}
We know from \cite{Kac,BLV} that $(I-Q) \geq 0$ and $(I-Q)h=0 \Leftrightarrow h$ is radial. Each $(I-P_k)$ is a projection with kernel precisely the subspace of functions in $\mathcal{X}_N$ that are independent of $v_k$. The only function in $\mathcal{X}_N$ that belongs to the kernel of $\sum_{k=1}^m (I-P_k)$ and is also radial is $0$. Hence, the Lemma is proved.
\end{proof}

The spectral gap of the operator $\el_{N,m}$ is defined as:

\begin{equation} \label{E:gapdef}
\Delta_{N,m} := \inf \{\<h,\el_{N,m}[h]>: ||h||=1, h \in \mathcal{X}_N\} \,\, .
\end{equation}

Lemma \ref{L:ker} implies that initial states in $\mathcal{X}_N$ decay to equilibrium at an exponential rate $\Delta_{N,m}$. 

\begin{xrem}
Gaussian states of temperature greater than twice the temperature of the heat bath cannot be represented by a function $h \in \mathcal{X}_N$.
\end{xrem}

The observation that $\el_{2,1}$ is simply a linear combination of two projections ($Q \equiv Q_{12}$ is an orthogonal projection onto radial functions in $\mathbb{R}^2$) lets us compute the whole spectrum in this case. This is done in Appendix \ref{A:sp2}. We see that the spectral gap is the lower root of the quadratic $x^2-(2\lambda+\mu)x+\lambda \mu$:

\begin{equation} \label{E:gap2}
\Delta_{2,1}:=\frac{(2\lambda+\mu)-\sqrt{4\lambda^2+\mu^2}}{2}
\end{equation}
with gap eigenfunction \[\frac{2\lambda}{2\lambda+\mu-\Delta_{2,1}} H_2(v_1) + \frac{2\lambda}{2\lambda - \Delta_{2,1}} H_2(v_2) , \]

where $H_2$ is the monic Hermite polynomial (with weight $\gamma$) of degree $2$. 

For general $N,m$, we have the following theorem:

\begin{thm} \label{T:gap}
Assume $\lambda, \mu >0$. Then 

\begin{equation}\label{E:est1}\frac{m}{N-1}\Delta_{2,1} \leq \Delta_{N,m} \leq \frac{m}{N-1} \frac{2 \lambda \mu  }{\mu+\lambda}\,\, .
\end{equation}
\end{thm}

\begin{proof}

The proof is based on an inductive argument that follows in essence the one in \cite{CCL} in which the spectral gap of the Kac model is computed exactly. We first prove the following claim for $1\leq m < N$:
 
\begin{equation} \label{E:claim}
\Delta_{N,m} \geq \frac{N-m-1}{N-1} \Delta_{N-1,m} + \frac{m}{N-1} \Delta_{N-1, m-1} \,\, .
\end{equation}

We let $\el_{N,m}^{(k)}$ be the evolution operator $\el_{N,m}$ with the $k^{th}$ particle removed:

\[ \el_{N,m}^{(k)} =  \frac{(N-1)\lambda} {\binom{N-1}{2}}\sum_{ \begin{array}{c} i{<}j\\ i,j \neq k \end{array}}^N (I - Q_{ij}) + \mu \sum_{ \begin{array}{c} l{=}1 \\ l {\neq} k\end{array}}^m (I- P_l). \]

\begin{rem} \label{R:induc}
$\el_{N,m}^{(k)}$ is also self-adjoint in $L^2(\mathbb{R}^N,\gamma)$, and will have $m$ or $m-1$ thermostats in it, depending on whether $k>m$ or $k \leq m$, respectively. Also, the coefficient of the Kac term corresponds to collisions among $N-1$ particles.
\end{rem}
 
Next we show that

\begin{equation} \label{E: motor}
 \el_{N,m} = \frac{1}{N-1} \sum_{k=1}^N \el_{N,m}^{(k)}.
\end{equation}

\noindent This follows, since

\begin{eqnarray*} 
             \sum_{k=1}^N \el_{N,m}^{(k)} & = & \sum_{k=1}^N 
\left( \frac{2 \lambda}{N-2} \sum_{\begin{array}{c} i<j,\\ i,j \neq k \end{array} }^N (I- Q_{ij}) + \mu \sum_{\begin{array}{c} l=1\\ l \neq k \end{array}}^m (I- P_l)  \right ) \\
                      & = & 2\lambda \sum_{i<j}^N (I - Q_{ij}) + (N-1) \mu \sum_{l=1}^m (I - P_l)\\
                      &= & (N-1) \el_{N,m}.
\end{eqnarray*}

Then 

\begin{equation} \label{E:conn}
\< h,\el_{N,m}[h]> = \frac{1}{N-1} \sum_{k=1}^N{\<h, \el_{N,m}^{(k)}[h]>}
\end{equation}

At this point, we want to introduce the gaps $\Delta_{N-1,m}$ and $\Delta_{N-1,m-1}$ for $N-1$ particles into the right hand side; for this, we will need the functions to be orthogonal to $1$ in the space $L^2(\mathbb{R}^{N-1},\gamma(\hat{v}_k))$, where $\gamma(\hat{v}_k)$ is the Gaussian $\gamma$ with the variable $v_k$ missing. To this end, we define the projections 
\[\pi_k[h]:=\int h \gamma(\hat{v_k})\, dv_1\dots dv_{k-1} dv_{k+1} \dots d{v_N}\]

and write, for each $k$, $\<h, \el_{N,m}^{(k)}[h]> = \<(h-\pi_k h),\el_{N,m}^{(k)}(h-\pi_k h)>$. This holds because the range of the projection $\pi_k$ is exactly the kernel of $\el_{N,m}^{(k)}$, and the operator $\el_{N,m}^{(k)}$ is self-adjoint. Thus, from \eqref{E:conn},

\[
\Delta_{N,m} = \frac{1}{N-1} \inf \sum_{k=1}^{N}{\<(h-\pi_k h),\el_{N,m}^{(k)}(h-\pi_k h)>}\,
\]

where the infimum is over $h \in \mathcal{X}_N$, $||h||=1$ as per the definition of the spectral gap. Since $(h-\pi_k h)$ is orthogonal to the constant function $1$ in $L^2(\mathbb{R}^{N-1},\gamma(\hat{v}_k))$ by construction, we use the definition of the spectral gap to write 

\begin{eqnarray*} 
\Delta_{N,m} &\geq& \frac{1}{N-1} \inf \left(\sum_{k=m+1}^N \Delta_{N-1,m}(||h-\pi_k h||^2) + \sum_{k=1}^m \Delta_{N-1,m-1}(||h-\pi_k h||^2)  \right) \text{ (by Remark \ref{R:induc})}\\
& = & \frac{1}{N{-}1} \inf \left( \Delta_{N-1,m} \sum_{k=m{+}1}^N (\vert\vert h \vert\vert ^2 - ||\pi_k h\vert\vert^2) + \Delta_{N-1,m-1} \sum_{k=1}^m (\vert\vert h||^2 - ||\pi_k h\vert\vert^2)\right)\\ & \geq & \frac{N-m}{N-1} \Delta_{N-1,m} + \frac{m}{N-1} \Delta_{N-1,m-1} - \frac{1}{N{-}1} \max\{\Delta_{N-1,m},\Delta_{N-1, m-1}\} \sup\sum_{k=1}^N \vert\vert \pi_k h\vert\vert^2 \ , 
\end{eqnarray*}

where we have used symmetry among $1,...,m$ and $m+1,...,N$ and the fact that the infimum is over functions with norm $1$.

First, we note that $\Delta_{N-1,m} \geq \Delta_{N-1,m-1}$ since $(I-P_m) \geq 0$. Next, $\sup\{\sum_{k=1}^N \vert\vert \pi_k h \vert\vert^2, h\in \mathcal{X}_N\}$ equals $\sup_{\mathcal{X}_N} \< h, \sum_{k=1}^N \pi_kh>$. Since $\{\pi_k\}_1^N$ is a collection of commuting projection operators, $\sum_{k=1}^N \pi_k$ is a projection and the supremum is $1$.

We then get
\[
\Delta_{N,m} \geq  \frac{N{-}m}{N{-}1} \Delta_{N{-}1,m} + \frac{m}{N{-}1} \Delta_{N{-}1,m{-}1} - \frac{1}{N-1} \Delta_{N-1,m},
\]

which implies claim \eqref{E:claim}.



We now prove the first inequality in Theorem \ref{T:gap}. The region of interest is $\{ (N,m): 1\leq m \leq N-1\}$. We will use induction on $N \geq 2$.

\begin{itemize}
\item The base case $N=2$, $m=1$ is the trivial statement $\Delta_{2,1}\geq \Delta_{2,1}$. 

\item Now suppose 
\begin{equation} \label{E:indu}
\Delta_{N,m} \geq \Delta_{2,1} \frac{m}{N-1}
\end{equation}
for all $m$ such that $1\leq m \leq N{-}1$. To show that $\Delta_{N+1,m} \geq \Delta_{2,1} \frac{m}{N}$ for all $m$ such that $1 \leq m \leq N$, consider the following two cases:

  \begin{itemize}
          \item $m=1$: We need to show that $\Delta_{N+1,1} \geq \frac{\Delta_{2,1}}{N}$. From \eqref{E:claim}, we deduce that \[\Delta_{N+1,1} \geq \frac{N-1}{N} \Delta_{N,1} + \frac{1}{N}\Delta_{N,0} = \frac{N-1}{N} \Delta_{N,1}\; .\]
In the above, we have $\Delta_{N,0}=0$ because when none of the particles are thermostated, the ground-state is degenerate (any radial function in $\mathbb{R}^N$ is an equilibrium for the Kac part). Applying \eqref{E:indu} with $m=1$ then completes the proof of this case.

          \item $1<m \leq N$:         
\begin{align*}
\Delta_{N+1,m} &\geq \frac{N{-}m}{N} \left( \frac{m \Delta_{2,1} }{N-1}\right)  +   \frac{m}{N} \left(\frac{(m{-}1) \Delta_{2,1}}{N-1}  \right) \text{ (using \eqref{E:claim} and \eqref{E:indu})}\\ &= \Delta_{2,1} \frac{m}{N(N-1)} (N-m + m -1) = \Delta_{2,1} \frac{m}{N}
\end{align*}

   \end{itemize}
\end{itemize}
This proves the first inequality in \eqref{E:est1}. We prove second inequality in \eqref{E:est1}, by finding an upper bound proportional to $\frac{ m }{N-1}$, for $\Delta_{N,m}$. This can be done by finding a (possibly crude) upper bound on
the eigenvalues of $\el_{N,m}$ on the space of second degree Hermite polynomials with weight $\gamma$. This space is invariant under $\el_{N,m}$ and its action on it with basis $\{ \sum_{k=m+1}^N H_2(v_k),\sum_{k=1}^m H_2(v_k) \}$ can be described by the following matrix (as mentioned before, this is related to the evolution of kinetic energy of the system).
We use the identities $Q_{ij}H_2(v_i) = (H_2(v_i)+H_2(v_j))/2$ and $Q_{ij}H_2(v_k) = H_2(v_k)$ for $i,j \neq k$ in obtaining the entries.

\[
\begin{pmatrix}
\frac{\lambda m}{N-1} & \frac{-\lambda m}{N-1} \\
-\frac{\lambda(N-m)}{N-1} & \frac{\lambda (N-m)}{N-1} + \mu
\end{pmatrix}.
\]

Its smallest eigenvalue is $\frac{1}{2}(\mu + \frac{N\lambda}{N-1}) \bigl( 1-\sqrt{1- \frac{4 m \lambda \mu}{N-1} \frac{1}{(\mu + \frac{N\lambda}{N-1})^2}} \bigr)$. Hence, by definition of the gap,

\[\Delta_{N,m} \leq \frac{1}{2}(\mu + \frac{N\lambda}{N-1}) \bigl( 1-\sqrt{1- \frac{m }{N-1} \frac{4\lambda \mu}{(\mu + \frac{N\lambda}{N-1})^2}} \bigr) \]

For $N$ large enough, we can write 

\[\Delta_{N,m} \leq \frac{1}{2}\left(\mu + \frac{N\lambda}{N-1}\right)\frac{m}{N-1} \frac{4\lambda \mu}{(\mu + \frac{N\lambda}{N-1})^2} \]

or 

\[ \Delta_{N,m} \leq \frac{m}{N-1} \frac{2\lambda \mu}{\mu + \lambda}\]

\end{proof}

Thus, as we are close to equilibrium, $h \to 0$ in $L^2(\mathbb{R}^N,\gamma)$ at an exponential rate $\Delta_{N,m}$, which for large $N$, is proportional to the fraction of thermostated particles. 

\section{van Hove Limit} \label{S:vh} 

In this section, we relate the strong and weak thermostats by studying the two-particle system ($N=2$, $m=1$) described by eq.~\eqref{E:mast}: 

\begin{equation} \label{E:hvh2}
\frac{\partial f^\lambda}{\partial t} = - 2\lambda (I-Q_{12})f^{\lambda} - \mu (I-R_1)f^\lambda =: -\mathcal{G}^\lambda f^\lambda \,\, ,
\end{equation}
 
Here the superscript makes it explicit that the solution depends on $\lambda$.
 
Particle $2$ interacts through the Kac collision with Particle $1$, which is given the Gaussian distribution $g(v)=\sqrt{\frac{1}{2\pi}}e^{-\frac{v^2}{2}}$ at random times due to the action of the strong thermostat $R_1$. We increase the rate $\mu$ at which the strong thermostat acts relative to the rate of the Kac collision $2\lambda$. This can be achieved by  increasing the time scale of the Kac operator $\frac{1}{2\lambda}\rightarrow \infty$ and sampling at longer time intervals $\displaystyle \tau:= t \lambda$. Thus, the thermostat, operating on a much smaller time-scale, becomes powerful in the limit. The result is that by passing through a van Hove (weak-coupling, large time) limit \cite{DA} of this system, Particle $2$ gets thermostated ``weakly", via its interaction with Particle $1$ whose distribution is essentially always $g(v)$.

We are interested in the evolution of $ \tilde{f}^\lambda(v_1, v_2,\tau):= f^\lambda(v_1,v_2,\frac{\tau}{\lambda})$ in the limit $ \lambda \to 0 $. Here $f^\lambda(v_1,v_2,t)$ satisfies \eqref{E:hvh2} above. The equation satisfied by $\tilde{f}^{\lambda}(v_1,v_2,\tau)$ is then:

\begin{equation} \label{E:til}
\frac{\partial \tilde{f}^{\lambda}}{\partial \tau} = -2(I-Q_{12}) \tilde{f}^{\lambda} - \frac{\mu}{\lambda}(I - R_1)\tilde{f}^{\lambda} =: -\frac{\mathcal{G}^\lambda}{\lambda} \tilde{f}^\lambda
\end{equation}

We have the following theorem, which states that the diagram in Figure $1$ commutes. 

\begin{thm} \label{T:vh}\end{thm}

Let $\tilde{f}^{\lambda}$ satisfy eq.~\eqref{E:til} with initial condition $\tilde{f}^{\lambda}(v_1,v_2,0)=\phi(v_1,v_2) \in L^1(\mathbb{R}^2)$. Then for $\tau>0, \, \displaystyle
\lim_{\lambda \to 0} \tilde{f}^{\lambda}=:g(v_1)\tilde{f}(v_2,\tau)$ exists in $L^1(\mathbb{R}^2)$, where $\tilde{f}$ satisfies the equation

\begin{equation}\label{E:vh}
 \frac{\partial \tilde{f}}{\partial \tau} = -2(I-U_2) \tilde{f} 
\end{equation}

together with the initial condition $\tilde{f}(v_2,0)=\frac{R_1 \phi(v_1,v_2)}{g(v_1)}$. $U_2$ is the weak thermostat \eqref{E:weak} acting on $v_2$. 

\begin{figure}[ here]
\begin{center}
\begin{tikzpicture}
  \matrix (m) [matrix of math nodes,row sep=3em,column sep=4em,minimum width=2em]
  {
     \phi(v_1, v_2) & \tilde{f}^{\lambda}(v_1, v_2, \tau) \\
     g(v_1)\tilde{f}(v_2,0) & \mbox{  }g(v_1)\tilde{f}(v_2,\tau)   \\};
  \path[-stealth]
    (m-1-1) edge node [left] {$R_1$} (m-2-1)
            edge node [above] {$e^{-\frac{\tau}{\lambda} \mathcal{G}^{\lambda}} $} (m-1-2)
    (m-2-1) edge node [below] {$e^{-2\tau(I-U_2)} $} (m-2-2)
    (m-1-2) edge node [right] {$\lambda\rightarrow 0$} (m-2-2);
\end{tikzpicture}
\caption{van Hove Limit}
\end{center}
\end{figure}

\begin{proof}
We can write $e^{-\frac{\mu}{\lambda} \tau(I-R_1)} = I + (I- R_1)(e^{-\mu \tau/\lambda}- 1)$ because $(I-R_1)$ is idempotent. This implies that 

\begin{equation}\label{E:vhbound}
\vert\vert e^{-\frac{\tau \mu}{\lambda} (I-R_1)}- R_1 \vert\vert_{1} =  e^{-\mu \frac{\tau}{\lambda}} \vert\vert I - R_1\vert\vert_1\leq 2e^{-\mu \frac{\tau}{\lambda}}.
\end{equation}

\vspace{0.5cm}

For each $\lambda$, the operators in $\frac{1}{\lambda}\mathcal{G}^\lambda$ are bounded. Thus, the Dyson expansion (the infinite series version of the Duhamel formula) corresponding to the evolution in \eqref{E:til} gives $ e^{-\frac{\tau}{\lambda} \mathcal{G}^\lambda} \phi   =  \sum_{k=0}^\infty b_k(\phi)$ where

\begin{eqnarray*}
       b_0 (\phi) & = & e^{-\frac{\mu}{\lambda} (I-R_1) \tau} \phi, \\
       b_1 (\phi) & = & \int_{t_1=0}^\tau e^{-\frac{\mu}{\lambda} (I-R_1) (\tau-t_1)} [-2(I- Q_{12})]  e^{-\frac{\mu}{\lambda} (I-R_1) t_1} \phi\,\,dt_1 \mbox{, and}\\
    b_k(\phi) & = & \int_{ \{0\leq t_k \leq \dots t_1 \leq \tau\}} e^{-\frac{\mu}{\lambda} (I-R_1) (\tau-t_1)}[-2(I- Q_{12})] e^{-\frac{\mu}{\lambda} (I-R_1) (t_1-t_2)} \dots [-2(I- Q_{12})] e^{-\frac{\mu}{\lambda} (I-R_1) (t_k)}\phi\,d\vec{t}\\     
\end{eqnarray*}

\noindent Using \eqref{E:vhbound} and the identity $R_1 Q_{12} R_1 = R_1 U_2 = U_2 R_1$, we show that $\forall k$, $b_k(\phi)$ converges to 

\[\int_{\{0\leq t_k \leq \dots t_1 \leq \tau\}} R_1 [-2(I- Q_{12})]R_1 \dots [-2(I- Q_{12})] R_1 \phi \,\, d\vec{t}=  \frac{1}{k!}\left( -2(I-U_2)\right)^k (R_1 \phi)\] 

in $L^1$ as $\lambda\rightarrow 0$.

Finally, we use the fact that for each $u\geq 0$, $\vert\vert e^{-\frac{\mu}{\lambda} (I - R_1) u} \phi \vert\vert_1  = \vert\vert \phi \vert\vert_1$ 
and $\vert\vert (I-Q_{12}) \phi \vert\vert_1 \leq 2 \vert\vert \phi \vert\vert_1$ so that 
$\vert\vert b_k(\phi) \vert\vert\leq 4^k \int_{\{0\leq t_k \leq \dots t_1 \leq \tau\}} dt_1 \dots dt_k  \vert\vert \phi \vert\vert_1 = \frac{(4\tau)^k}{k!}\vert\vert \phi \vert\vert_1$, 
independently of $\lambda$. Therefore the dominated convergence theorem can be applied to give

\begin{eqnarray*}
 \lim_{\lambda\rightarrow 0} e^{-\frac{\tau}{\lambda} \mathcal{G}^\lambda} \phi & = & \lim_{\lambda\rightarrow 0} \sum_{k=0}^\infty b_k(\phi) = \sum_{k=0}^\infty \lim_{\lambda\rightarrow 0} b_k(\phi)\\
       & = & \sum_{k=0}^\infty (-2(I-U_2))^k \frac{\tau^k}{k!} (R_1 \phi) = e^{-2(I - U_2)\tau} (R_1 \phi).
\end{eqnarray*}

\end{proof}

The above proof can be generalized to give the following van Hove results for the $N$-particle case. We will use the statement ``the van Hove limit of $\{ A(\lambda): \lambda>0\}$ as $\lambda\rightarrow 0$
is $A^\ast$ with idempotent operator $B$'' to mean $\lim_{\lambda\rightarrow 0} e^{-\frac{\tau}{\lambda} A(\lambda)} \phi = e^{-\tau A^\ast} (B\phi) = B e^{-\tau A^\ast} \phi$ for all $\tau>0$ and all $\phi \in L^1$.

\begin{itemize}
 \item The van Hove limit of $\{ \lambda \sum_{j=2}^N ( I - Q_{1j}) + \mu (I - R_1)\}$, acting on $L^1(\RR^N)$, is $\sum_{j=2}^N(I-U_j)$ with idempotent operator $R_1$. 
 \item The van Hove limit of $\{N \lambda (I-Q) + \mu (I-R_1)\}$, acting on $L^1(\RR^N)$ is $\frac{2}{N-1}\sum_{j=2}^N (I-U_j) + \frac{N-2}{N-1} (N-1)(I- Q^{(1)})$ 
with idempotent operator $R_1$. Here $Q^{(1)} = \frac{2}{N-2} \sum_{2 \leq i<j} (I-Q_{ij})$ is the Kac operator acting on particles $2, \dots, N$. 
\item Let $\alpha= \frac{2N-1}{N-1}$. The van Hove limit of $\{ \lambda \alpha 2N(I-Q_{(2N)}) + \mu \sum_{i=N+1}^{2N} (I-R_i)\}$, acting on $L^1(\RR^{2N})$,
is $N(I-Q_{(N)}) + \frac{2N}{N-1} \sum_{i=1}^N (I-U_i)$, with idempotent operator $R_{N+1} R_{N+2} \dots R_{2N}$. Here $Q_{(N)}$ and $Q_{(2N)}$ are the Kac operators acting on particles $v_1, \dots, v_{N}$ and $v_1, \dots, v_{2N}$ respectively.
\end{itemize}

The first two results show that having one strongly thermostated particle is sufficient to``weakly" thermostat each particle colliding with it in the van Hove limit. The strength of this thermostat will be $O(\frac{1}{N})$ under the usual Kac collision unless $\sim N$ strongly thermostated particles are used, as in the third result.

\begin{rem}\label{R:BLV}
The third result shows that up to a constant in the thermostat terms, it is possible to obtain the model in \cite{BLV} as a van Hove limit of models in which half the particles are strongly thermostated.
\end{rem}

\section{Approach to Equilibrium in Entropy} \label{S:entr}

In this section, we study the behavior of the relative entropy functional

\begin{equation} \label{E:relent}
S(f|\gamma):= \int f\log \frac{f}{\gamma} d\V
\end{equation}

under the evolution~\eqref{E:mast}. This is a standard way to track the approach to equilibrium since it satisfies $S(f|\gamma) \geq 0$ and $S(f|\gamma)=0 \Leftrightarrow f=\gamma$. For our model, we show below that $S(f(.,t)|\gamma) \to 0$ as $t \to \infty$, provided the initial distribution $f(.,0)$ has finite relative entropy. 

Set $f=\gamma h$ (this is slightly different from the ground-state transformation \eqref{E:gs}), and restrict to $h \geq 0$, $\int h \gamma d\V=1$. The evolution equation obeyed by $h(\V,t)$ is eq.~\eqref{E:hmast}, which we restate below:

\[\frac{\partial h}{\partial t} = N \lambda (Q-I) h + \mu \sum_{k=1}^m(P_k-I)h = -\el_{N,m} h \,\, . \]

The relative entropy then becomes $\int h \log h \, \gamma \, d\V$, which we denote by $S(h)$ (overloading the notation) for the remainder of this section.

Now, 
\[\frac{dS}{dt} \, = \int \frac{\partial h}{\partial t} \log h \, \gamma d\V + \int \frac{h}{h} \frac{\partial h}{\partial t} \, \gamma d\V  \,= \int \frac{\partial h}{\partial t} \log h \, \gamma d\V \, ,\]

where the second term vanishes because the normalization $\int h \gamma \, d\V =1$ is preserved by the evolution. Hence,
\[
\frac{dS}{dt} =\int \left(N \lambda (Q-I) h + \mu \sum_{k=1}^m(P_k-I)h\right) \log h \, \gamma d\V \, .
\]

We know (from \cite{Kac}) that $\int N (Q-I) h \log h \, \gamma d\V \leq 0$. Also, 
\begin{align*}
\int P_k h \log h \, \gamma d\V &= \int P_k h \,\, P_k(\log h) \gamma d\V \,\,\,\, \text{(by self-adjointness of $P_k$ as observed in Section~\ref{S:L2})}\\
& \leq \int (P_k h) \log(P_k h) \gamma d\V \,\,\,\, \text{(by concavity of $\log$ and averaging property of $P_k$)}\\
& \leq \int h \log h \, \gamma d\V \,\,\,\, \text{(by convexity of $x\log x$)}
\end{align*}

Thus $\frac{dS}{dt} \leq 0$. The following theorem indicates how fast the relative entropy decays under the evolution.

\begin{thm} \label{T:duh}
Assume $1 \leq m <N$ and let $h(\V,t)$ be the solution of \eqref{E:hmast}.
Then we have that 
 \begin{equation} \label{eq:main}
  S(h(\V,t)) \leq \left(-\frac{\delta_- e^{-\delta_+t}}{\delta_{+}-\delta_{-}} + \frac{\delta_+ e^{-\delta_-t}}{\delta_{+}-\delta_{-}} \right)S(h(\V,0))
 \end{equation}

where $\delta_{\pm} \equiv \delta_{\pm}(N,m)=\left(\frac{N\lambda+\mu}{2} \pm \frac{1}{2}\sqrt{(N\lambda+\mu)^2-4m\lambda \mu/(N-1)}\right)$. 
\end{thm}

We first state a few observations on the above bound. Let us define 
 
 \[D(t):=-\frac{\delta_- e^{-\delta_+t}}{\delta_{+}-\delta_{-}} + \frac{\delta_+ e^{-\delta_-t}}{\delta_{+}-\delta_{-}} \, .\] 

As expected, $D(t)$ is identically equal to $1$ when $\lambda$ or $\mu$ is $0$. For $\lambda, \mu >0$, $\displaystyle\lim_{t \to \infty} D(t) = 0$, $D(t)$ is equal to $1$ at $t=0$ and it is a decreasing function of $t>0$. The last claim can be seen by computing

\begin{equation} \label{E:entr0}
\frac{dD}{dt} = \frac{\delta_- \delta_+}{\delta_{+}-\delta_{-}} \bigl( e^{-\delta_+ t}-e^{-\delta_- t}\bigr) \leq 0
\end{equation}

since $\delta_-<\delta_+$. For large $t$, the dominant term in the bound \eqref{eq:main} is $e^{-\delta_- t}$, and for large $N$, $\delta_- \sim \frac{m\lambda \mu}{(N-1)(N\lambda+\mu)}$. Hence, we obtain an eventually exponential decay of relative entropy through this bound, albeit with decay constant $\sim \frac{m}{N^2}$.

In this paragraph, we make a few remarks about the bound for the special case $N=2, m=1$. Observe that $\delta_-(2,1) =\Delta_{2,1}$ is the spectral gap of $2\lambda(I-Q)+\mu(I-P_1)$ (see \eqref{E:gap2}). As an aside, note that this is in accordance with \eqref{E:entbd}. Upon making the transformation $(\mu,\lambda) \rightarrow (\frac{\mu}{\lambda},1)$ corresponding to the van Hove limit (see eq.~\eqref{E:til}), we obtain $D(t) \rightarrow e^{-t}$ as 
 $\lambda \rightarrow 0$. This is exactly the optimal entropy production bound \eqref{E:entrweak} for the weak thermostat (Note: the weak thermostat here appears with a factor of $2$, owing to the $2\lambda$ term).

The Theorem is proved as follows: we write $h(\V,t)$ explicitly in terms of the exponential of the generator of the evolution, expand the latter using the Dyson series and use the convexity of the entropy. We exploit the entropic \emph{contraction} of terms of the form $P_jQ$ in the expansion. These steps will yield a non-trivial bound for the entropy at time $t$ in terms of the initial entropy.

The following lemmas build up to the evolution operator $e^{-\el_{N,m}t}$ in steps. For instance, Lemma \ref{L:half} bounds some of the terms obtained by decomposing the Kac operator in the expression $S(P_1Qh)$. Throughout, we assume that $h \in L^1(\mathbb{R}^N,\gamma)$ and $h \geq 0$.

\begin{lem} \label{L:half}
We have
\[ \sum_{j=2}^N S(P_1 Q_{1j}h) \leq \bigl((N-1) - \frac12 \bigr) S(h)\]
\end{lem}

\begin{proof}
In the following proof, we will apply the continuous version of Han's inequality \cite{Han} (this also follows from the Loomis-Whitney inequality \cite{LW}) for the entropy rewritten to suit our situation:

\begin{equation} \label{E:Han}
\sum_{j=1}^N S(P_jh) \leq (N-1)S(h)
\end{equation}

Note that if $h$ is symmetric in its arguments, this amounts to saying that for each $j=1,..,N$, 
\begin{equation} \label{E:Hansym}
S(P_j h) \leq \frac{N-1}{N} S(h)
\end{equation}
For $j>1$, 
\begin{align*}
S(P_1Q_{1j}h) &= \int P_1 Q_{1j}h \log \bigl(P_1 Q_{1j}h\bigr) \; \gamma d\V \\
&=\int P_1 (\frac{Q_{1j}h}{P_1P_jh}) \log \bigl( P_1 (\frac{Q_{1j}h}{P_1P_jh}) \bigr) P_1P_jh \; \gamma d\V + \int P_1Q_{1j}h \log (P_1P_jh) \; \gamma d\V
\end{align*}
where we use that $P_1P_jh$ does not depend on $v_1$. Since the argument of the logarithm in the last term is also independent of $v_j$, we can integrate $P_1Q_{1j}h$ with respect to those variables and use that $\int P_1 Q_{1j}h \;g(v_1)g(v_j)dv_1 dv_j = \int h \;g(v_1)g(v_j)dv_1 dv_j = P_1P_jh$ to write:

\[S(P_1Q_{1j}h) = \int P_1 (\frac{Q_{1j}h}{P_1P_jh}) \log \bigl( P_1 (\frac{Q_{1j}h}{P_1P_jh}) \bigr) P_1P_jh \; \gamma d\V + \int P_1P_jh \log (P_1P_jh) \; \gamma d\V\]

Now, we apply the symmetric version of Han's inequality \eqref{E:Hansym} to $\frac{Q_{1j}h}{P_1P_jh}$ as a function of $v_1$ and $v_j$ to get:

\begin{align*}
S(P_1 Q_{1j}h) &\leq \frac{1}{2} \int \frac{Q_{1j}h}{P_1P_jh} \log \bigl(\frac{Q_{1j}h}{P_1P_jh} \bigr) P_1P_jh \; \gamma d\V + \int P_1P_jh \log (P_1P_jh) \; \gamma d\V \\
&= \frac{1}{2} S(Q_{1j}h) - \frac{1}{2} \int Q_{1j}h \log (P_1P_jh) \; \gamma d\V + \int P_1P_jh \log(P_1 P_j h) \; \gamma d\V \\
&=\frac{1}{2}S(Q_{1j}h) + \frac{1}{2}S(P_1P_jh)
\end{align*}

where, to get to the last step, we have used that $Q_{1j}$ is self-adjoint and $P_1P_j$ is independent of $v_1$ and $v_j$.

Now, summing these terms, and noting that $S(Q_{1j}h) \leq S(h)$ by the averaging property of $Q_{1j}$, we get
\[
\sum_{j=2}^N S(P_1 Q_{1j}h) \leq \frac{N-1}{2}S(h) + \frac{1}{2}\sum_{j=2}^N S(P_jP_1h) \,.
\]

We invoke Han's inequality \eqref{E:Han} on $P_1 h \equiv (P_1h)(v_2,...v_N)$, ie. $\sum_{j=2}^N S(P_jP_1h) \leq (N-2)S(P_1h) \leq (N-2)S(h)$ to complete the proof.
\end{proof}

\begin{lem} 
\[S(e^{\mu(P_1-I)t}Qh) \leq \left(1-\frac{1-e^{-\mu t}}{N(N-1)}\right) S(h) \, .\]
\end{lem}
\begin{proof}
\begin{align*}
S(e^{\mu(P_1-I)t}Qh) &= S(e^{-\mu t}Qh+(1-e^{-\mu t})P_1Qh) \;\;\;\;\; \text{  (since $P_1$ is a projection)} \\
&\leq e^{-\mu t} S(Qh) + (1-e^{-\mu t}) S(P_1Qh) \\
&\leq e^{-\mu t} S(h) + (1-e^{-\mu t}) \frac{1}{\binom{N}{2}} \sum_{i<j}S(P_1Q_{ij}h)\\
&= e^{-\mu t} S(h) + (1-e^{-\mu t}) \frac{1}{\binom{N}{2}} \bigl(\sum_{i<j, i,j \neq 1} S(P_1 Q_{ij}h) + \sum_{j=2}^N S(P_1 Q_{1j}h)\bigr) \\
&\leq e^{-\mu t} S(h) + (1-e^{-\mu t}) \frac{1}{\binom{N}{2}} \bigl(\sum_{i<j, i,j \neq 1} S(h) + (N-1 - \frac{1}{2})S(h)\bigr) \\
&= \left(1-\frac{1-e^{-\mu t}}{N(N-1)}\right) S(h),
\end{align*}

where we use Lemma \ref{L:half} in the last inequality. We use the convexity of the entropy and the averaging property of $P_1$ and $Q$ in the previous steps.
\end{proof}

\begin{lem} \label{L:2entr}
Let $1 \leq m < N$. Then
\begin{equation} \label{E:m}
S\left(\exp{\left(\mu \displaystyle\sum_{k=1}^m (P_k-I)t)\right)} Q h\right) \leq \left( 1-\frac{m(1-e^{-\mu t})}{N(N-1)} \right)S(h).
\end{equation}
\end{lem}

\begin{proof}
We prove the above by induction on $m$. The base case $m=1$ (and any $N>1$) was shown in the previous Lemma. We restrict to $\{(N,m):2 \leq m < N\}$ for the rest of the proof. Assume that the Lemma is true for $m-1$ (and any $N > m-1)$. To infer from this its validity for the case $m$ (and any $N>m$), we analyze below the entropy of $P_m \exp{\left(\mu \sum_{k=1}^{m-1}(P_k-I)t\right)}$, where we expand the Kac operator $Q$, split it into terms that contain $m$ and those that do not, and utilize the convexity of the entropy. 

\begin{align*}
S\left(P_m \exp{\bigl(\mu \sum_{k=1}^{m-1}(P_k-I)t\bigr)}Qh\right) &\leq (1-\frac2N)S\left(\frac{\exp{\bigl(\mu \sum_{k=1}^{m-1}(P_k-I)t\bigr)}}{\binom{N-1}{2}} \displaystyle\sum_{\substack{i<j \\ i,j \neq m}} Q_{ij}P_m h \right) \\
&+\frac{2}{N}S\left(\frac{\exp{\bigl(\mu \sum_{k=1}^{m-1}(P_k-I)t}\bigr)}{N-1} P_m \displaystyle\sum_{l \neq m} Q_{lm}h \right) \, .
\end{align*}

In the first term\footnote{This term is non-zero only when $N>2$, which is the case here.}, we also use the commutativity of $P_m$ with $Q_{ij}$ when neither $i$ nor $j$ equal $m$. Next, we treat the terms as follows: 
\begin{itemize}
\item Term 1: We apply the induction hypothesis for $m-1$, $N-1$ since $P_m h$ is a function of $N-1$ variables and $\binom{N-1}{2}^{-1} \displaystyle\sum_{\substack{i<j \\ i,j \neq m}} Q_{ij}$ is the Kac operator acting on $N-1$ variables.
\item Term 2: We  use the averaging property of $\exp{\bigl(\mu \sum_{k=1}^{m-1}(P_k-I)t}\bigr)$, convexity, and Lemma \ref{L:half}. 
\end{itemize}

We obtain 
\begin{align} \label{E:now}
S(P_m \exp{\bigl(\mu \sum_{k=1}^{m-1}(P_k-I)t\bigr)}Qh) &\leq (1-\frac2N)\left(1-(m-1)\frac{1-e^{-\mu t}}{(N-1)(N-2)}\right)S(h)  + \frac2N \frac{1}{N-1} (N- \frac32) S(h) \, .
\end{align}

Now starting with the left-hand side of \eqref{E:m} and using convexity plus the fact that $P_m$ is a projection, write
\begin{align*}
S\left(\exp{(\mu \displaystyle\sum_{k=1}^m (P_k-I)t)} Q h \right) &= S\left((e^{-\mu t}I+(1-e^{-\mu t})P_m)\exp{\bigl(\mu \sum_{k=1}^{m-1}(P_k-I)t\bigr)}Qh \right)\\
&\leq e^{-\mu t} S\left(\exp{\bigl(\mu \sum_{k=1}^{m-1}(P_k-I)t\bigr)}Qh \right) \\ &+ (1-e^{-\mu t}) S\left(P_m \exp{\bigl(\mu \sum_{k=1}^{m-1}(P_k-I)t\bigr)}Qh \right)
\end{align*}

Using the induction hypothesis for the case $m-1$, $N$ for the first term, and the bound \eqref{E:now} for the second term, the Lemma follows through some algebraic simplification.

\end{proof}

In the following, denote $A(t):=1-\frac{m(1-e^{-\mu t})}{N(N-1)}$. 

\begin{proof}[Proof of Theorem~\ref{T:duh}]

Expanding $e^{-\el_{N,m}t}$ using the Dyson series with $Q$ as the perturbation:
\begin{align*}
e^{N\lambda (Q-I)t+\mu \sum_k(P_k-I)t}&=e^{-N\lambda t} e^{N\lambda Qt+\mu \sum(P_k-I)t} \\
&= e^{-N\lambda t} \{e^{\mu \sum(P_k-I)t}+\int_0^t dt_1 e^{\mu \sum(P_k-I)(t-t_1)} \; N \lambda Q \; e^{\mu \sum(P_k-I)t_1} \\
&+ \int_0^t dt_1 \int_0^{t_1} dt_2\; e^{\mu \sum(P_k-I)(t-t_1)} \; N\lambda Q \; e^{\mu \sum(P_k-I)(t_1 - t_2)} \;N\lambda Q \;e^{\mu \sum(P_k-I)t_2} + ...\}
\end{align*}

Therefore, using the convexity of entropy, and Lemma~\ref{L:2entr},
\[S(h(.,t)) \leq e^{-N\lambda t} \left(1+N\lambda \int_0^t dt_1 A(t-t_1) +(N\lambda)^2 \int_0^t dt_1 \int_0^{t_1} dt_2 A(t-t_1)A(t_1-t_2) + ...\right) S(h(.,0)) \]
 \[= e^{-N\lambda t} \left(1+N\lambda (A*1) + (N\lambda)^2 (A*A*1) + ...\right) S(h(.,0)) \]
where $*$ is the Laplace-convolution operation.
Thus we have that
\begin{equation} \label{eq2}
 S(h(.,t)) \leq e^{-N \lambda t}\; \varphi(t) S(h(.,0))
\end{equation}

where $\varphi$ is defined through the series above. We compute $\varphi(t)$ using its Laplace transform $\tilde{\varphi}(s)$.
Then:
\[\tilde{\varphi}(s)= \frac{1}{s} \; \sum_{k=0}^{\infty}{(N \lambda \; \tilde{A}(s))^k} \]
where $\tilde{A}(s)=\frac{1}{s}-\frac{m}{N(N-1)}(\frac{1}{s}-\frac{1}{s+\mu})$ is the Laplace transform of $A(t)$.

Summing the geometric series (the sum converges if we assume, for instance, that $\tilde{\varphi}(s)$ is defined on the domain $s>N\lambda$), 

\[\tilde{\varphi}(s)=\frac{s+\mu}{s^2+(\mu-N\lambda)s-N\mu \lambda (1-\frac{m}{N(N-1)})}\]

The inverse Laplace transform of the above is
\[-\frac{\delta_- e^{(N\lambda - \delta_+)t}}{\sqrt{(N\lambda+\mu)^2-4m\lambda \mu /(N-1)}} + \frac{\delta_+ e^{(N\lambda-\delta_-)t}}{\sqrt{(N\lambda+\mu)^2-4m\lambda \mu /(N-1)}} \]

Now we invoke the uniqueness of the Inverse Laplace Transform: No two piecewise continuous, locally bounded functions of exponential order can have the same Laplace transform (see e.g. \cite{Churchill}). Since $\varphi(t)$ (see eq.~\eqref{eq2}) belongs to this space, we get

\[\varphi(t)=  -\frac{\delta_- e^{(N\lambda - \delta_+)t}}{\sqrt{(N\lambda+\mu)^2-4m\lambda \mu /(N-1)}} + \frac{\delta_+ e^{(N\lambda-\delta_-)t}}{\sqrt{(N\lambda+\mu)^2-4m\lambda \mu /(N-1)}}\]

Plugging this into \eqref{eq2}, we obtain the desired result \eqref{eq:main}.

\end{proof}

\begin{xrems} \
\begin{itemize}
\item From~\eqref{E:entr0}, one notices that $\frac{dD}{dt}|_{t=0}=0$. This implies, in particular, that Theorem \ref{T:duh} does not give us a bound like \eqref{E:prod} on the entropy production. This results from the fact that the significant bounds used in the proof, from Lemma \ref{L:half}, required the presence of the second-order term $\sum_k (P_k-I)Q$. Note that $\frac{d^2D}{dt^2}|_{t=0}<0$.
\item The main bound (Lemma \ref{L:half}) was obtained by estimating terms of the form $S(P_1 Q_{1j}h)$, and we ignored any possible contribution from many other terms e.g. $S(Q_{ij}Q_{kl}h)$. Thus, there may be scope for a better bound.
\item In particular, we hope to obtain an entropy decay rate that scales as $\frac{m}{N}$ (as we had for the spectral gap). We were able to obtain a decay rate scaling as $\frac{1}{N}$ for a modified model: a system of $N$ particles where one of them is thermostated (through a Maxwellian thermostat) and the Kac collision interaction is replaced by the (much stronger) projection onto radial functions. Thus, the role of the Kac interaction in the equilibration process needs to be better understood.
\end{itemize}
\end{xrems}

Finally, we demonstrate why it is not easy to find an entropy production bound in our problem. Consider the case $N=2$, $m=1$ with $\lambda=\frac12$, $\mu=1$. Here, one could write \[\frac{dS(h)}{dt} = \int{P_1h \log h \gamma d\V} + \int{Qh \log h \gamma d\V} -2S(h)\] \[ \leq \int{P_1h \log P_1h \gamma d\V} + \int{Qh \log Qh \gamma d\V} -2S(h) \; .\] 

We use in the last step that $P_1$, $Q$ are projections and $\log x$ is concave. Bounding this from above by $-k S(h)$ (for some $k>0$) would be sufficient to obtain an entropy production bound. This idea has worked, e.g., for a sum of mutually orthogonal projections like strong thermostats acting on different particles. However, in our case, we can find, for every $\epsilon>0$, a density $h_{\epsilon}$ such that
\[\frac{\int{P_1h_{\epsilon} \log P_1h_{\epsilon} \gamma d\V} + \int{Qh_{\epsilon} \log Qh_{\epsilon} \gamma d\V}}{S(h_{\epsilon})} \geq 2-\epsilon \]

The idea is to take $h$ proportional to the characteristic function of the set $[-a,a] \times [R-a,R+a]$. As $R \to \infty$, the ratio above asymptotically approaches the value $2$. The intuition behind this construction is that as $R \to \infty$, $h$ is supported approximately in the intersection of the supports of $P_1h$ (a ``band'' of width $2a$ parallel to the $v_1$ axis) and $Qh$ (an annulus around the origin). It is the tangential nature of this intersection that precludes the application of Han's inequality \cite{Han} to improve the bound $S(P_1h)+S(Qh) \leq 2S(h)$. We are not, however, ruling out the possibility of using a different method to obtain an entropy production bound.

\section{Conclusion}

Our results imply that if a macroscopic fraction of particles is thermostated, the kinetic energy and the $L^2$ norm decay exponentially to their respective equilibrium values at a rate \emph{independent} of $N$.  However, our entropy bound \eqref{eq:main} yields a decay rate that vanishes as $\frac{1}{N}$ in the thermodynamic limit. Hence, at least under a suitable class of initial conditions, we think it should be possible to improve \eqref{eq:main} to reflect the physical situation.

The question of entropy production at $t=0$ (and any $N$) remains unsettled. The bound \eqref{eq:main} does not preclude the possibility of zero entropy production at time $0$. However, we do not know if it actually occurs in the model for some initial conditions.

One could wonder how the notion of propagation of chaos (which was the main motivation behind the formulation of the Kac model) adapts to our situation. When $m$ is finite, the coupling to the heat bath becomes insignificant in the thermodynamic limit. On the other hand, when $m=\alpha N$ for some $\alpha <1$, preliminary calculations indicate that in the limit, a coupled Boltzmann equation system should result. The \emph{Stosszahlansatz} needs to be reformulated in a precise manner to account for different distributions of the thermostated and the non-thermostated particles. Moreover, generalizations of our model could bring about connections to previously studied thermostated Boltzmann equations \cite{CLM}.

Lastly, the results in Section \ref{S:vh} suggest that it should be possible to extend our analysis to the case of systems partially coupled to the \emph{weak} thermostat.

\section*{Acknowledgement}
We are grateful to our advisors Federico Bonetto and Michael Loss for 
suggesting this topic, and for very fruitful discussions.
We also thank them for their help in Theorem \ref{T:gap} and Lemma \ref{L:half}. 

\appendix

\section{Appendix: Spectrum of Evolution Operator for $N=2,m=1$} \label{A:sp2}
We analyze the spectrum of the self-adjoint evolution operator $\el_{2,1}=2\lambda (I-Q)+\mu(I-P_1)$, in the space $L^2(\mathbb{R}^2,\gamma(\V)d\V)$, and deduce its spectral gap stated in \eqref{E:gap2}. For simplicity, we denote the operators $\el_{2,1}$ and $P_1$ by $\el$ and $P$.

Notice that $\el$ is a linear combination of two projections ($Q \equiv Q_{12}$ is an orthogonal projection onto radial functions in $\mathbb{R}^2$). The condition $\<h,1>=0$ corresponding to the normalization of $f = \gamma (1+h)$, the leads us to work in the space of Hermite polynomials $\{H_{\alpha}(v)\}_{\alpha=0}^{\infty}$ with weight $g(v)$. The space  of interest $\mathcal{X}_2$ is spanned by $\{K_{i,j}: i,j \in \mathbb{N}, (i,j) \neq (0,0) \}$, where $K_{i,j}:=H_i(v_1)H_j(v_2)$. Without loss of generality, we work with monic Hermite polynomials.

The action of $P$ is as follows:

\[ P K_{i,j} = \left\{
     \begin{array}{lr}
       0 & : i \neq 0\\
       K_{0,j} & : i=0
     \end{array}
   \right.\]

Since each term in $K_{i,j}$ is odd in either $v_1$ or $v_2$ when either $i$ or $j$ is odd, we have that $Q K_{i,j} =0$ when either $i$ or $j$ is odd. We deduce the action of $Q$ on $K_{2\alpha_1,2\alpha_2}$ from its action on $v_1^{2\alpha_1}v_2^{2\alpha_2}$ using the following Lemma from \cite{BLV}, which applies to $Q$ as it is a projection onto radial functions.

\begin{lem}\cite{BLV}
Let $A$ be a self-adjoint operator on $L^2(\mathbb{R}^N,\gamma(\V)d\V)$ that preserves 
the space $P_{2l}$ of homogeneous even polynomials in $v_1,...,v_N$ of degree $2l$. 
If
\[
A(v_1^{2\alpha_1}...v_N^{2\alpha_N})=\sum_{\sum \alpha_i = \sum \beta_i} c_{ 
\beta_1...\beta_N}v_1^{2\beta_1}...v_N^{2\beta_N} \ ,
\] 
we get
\[
A(H_{2\alpha_1}(v_1)...H_{2\alpha_N}(v_N))=\sum_{\sum \alpha_i = \sum \beta_i}
{c_{\beta_1...\beta_N}H_{2\beta_1}(v_1)...H_{2\beta_N}(v_N)} \ .
\]
\end{lem}
               
Let $n:=\alpha_1+\alpha_2$ and $\Gamma_{\alpha_1,\alpha_2}:=\dashint_0^{2\pi}{\cos^{2\alpha_1}{\theta} \sin^{2\alpha_2}{\theta}} d\theta=\frac{(2\alpha_1-1) !! (2\alpha_2-1)!!}{2^{\alpha_1+\alpha_2}(\alpha_1+\alpha_2)!}$, with the standard definition $(-1)!!=1$. Then we have

\[Q K_{i,j}= \left\{ \begin{array}{ll}
                0 & : i \text{ or } j \text{ odd} \\
                \Gamma_{\alpha_1,\alpha_2} \sum_{m=0}^n \binom{n}{m} K_{2m,2n-2m} & :i=2\alpha_1, j=2\alpha_2 
                \end{array} \right.\]

Now a case-by-case analysis, using the fact that $L_{2n}:=\text{Span}\{H_{2\alpha_1}(v_1)H_{2\alpha_2}(v_2):\alpha_1+\alpha_2=n\}$ are invariant subspaces for $\el$, yields the following for the spectrum of $\el$:

\begin{center}
\begin{tabular}{|c|c|} 
\hline
Eigenvalue & Eigenfunction \\
\hline
$2\lambda + \mu$ & $K_{i,j}$, $i$ or $j$ odd, $i \neq 0$ \\
                 & $\sum_{i=1}^n c_i K_{2i,2n-2i}$ where $\sum_{i=1}^n c_i \Gamma_{i,n-i}=0$\\
                 
\hline
$2\lambda$ & $K_{0,j}$, $j$ odd \\
\hline
$x^{\pm,n}$ & $\sum_{i=0}^n c^{\pm,n}_i K_{2i,2n-2i}$ and eq.~\eqref{E:cond} \\
\hline
\end{tabular} 
\end{center}

\begin{rem}
The first row corresponds to functions that belong to the kernels of both $Q$ and $P$, and the second row to functions that belong to the kernels of $Q$ and $I-P$.
\end{rem}

Here, \[x^{\pm,n}=\frac{(2\lambda+\mu) \pm \sqrt{(2\lambda+\mu)^2-8\lambda \mu(1-\Gamma_{0,n})}}{2}\]

and

\begin{equation} \label{E:cond}
c_0^{\pm,n}= \frac{2\lambda}{2\lambda - x^{\pm,n}} \text{ and } c_i^{\pm,n}= \frac{2\lambda \binom{n}{i}}{x^{\mp,n}}\text{ for } i \neq 0
\end{equation}

Using the fact that $\Gamma_{0,n}=\frac{1}{2\pi}\int_0^{2\pi}{\cos^{2n}{\theta} d\theta}$ is decreasing in $n$, it is easy to see that the smallest eigenvalue is $x^{-,1}$. The corresponding eigenfunction is $\frac{2\lambda}{2\lambda-x^{-,1}}K_{0,2} + \frac{2\lambda}{x^{+,1}}K_{2,0}$.

\section{Appendix: Entropy Bound Optimizer for the Weak Thermostat} \label{A:appopt}
 
In \cite{BLV}, the convexity of entropy is employed to show that if $f(v,t)$ satisfies

\[\frac{\partial f}{\partial t} = \eta(Uf - f)\]

where $U$ is the weak thermostat, then \eqref{E:kappa} holds true.

We remark here that if $\phi_\delta (v):=(1-\delta)M_x(v)+\delta M_y(v)$,  where $x=\frac{1}{\beta(1-\delta)}$, $y=\frac{1}{\beta\delta}$ and $M_a(v)=\frac{1}{\sqrt{2\pi a}}e^{-v^2/2a}$, then 

\[\lim_{\delta \to 0} \frac{1}{S(\phi_\delta)}\frac{dS}{dt}(\phi_\delta) \geq -\frac{\eta}{2}\]

thereby showing that \eqref{E:kappa} is an optimal bound. $\phi_\delta$ is a convex combination of Maxwellians, one of which approaches the distribution of the heat bath $M_{\frac{1}{\beta}}$ and the other corresponds to a very high energy distribution (albeit with a vanishing weight) as $\delta \to 0$. These types of functions have been used in \cite{bobcerc,CCLRV,amit} as examples of distributions that are away from equilibrium (in the sense of the entropy) and yet have vanishingly low entropy production (in magnitude). 

\bibliographystyle{abbrv}
\bibliography{nonequi}

\end{document}